\documentclass[english,letterpaper,11pt]{article}
\usepackage[utf8]{inputenc}
\usepackage{amsmath,amssymb,amsthm}
\usepackage{mathrsfs}
\usepackage{graphicx}
\usepackage{appendix}
\usepackage{fullpage}
\usepackage[skip=2pt]{parskip}
\usepackage{setspace}
\usepackage{multirow}
\usepackage{multicol}
\usepackage{booktabs}
\usepackage{bm}
\usepackage{natbib}
\usepackage{thmtools}
\usepackage{enumitem}
\usepackage{hyperref}
\usepackage{color}
\usepackage[dvipsnames]{xcolor}
\definecolor{darkblue}{rgb}{0, 0, 0.5}
\hypersetup{
    colorlinks = true,
    citecolor = darkblue,
    urlcolor = darkblue,
    linkcolor = darkblue
}
\usepackage{cleveref}
\crefname{thm}{Theorem}{Theorems}
\crefname{assumption}{Assumption}{Assumptions}
\usepackage{enumitem}
\usepackage{authblk}

\newcommand{\E}{\operatorname{\mathbb{E}}}

\newcommand{\Prb}{\operatorname{\mathbb{P}}}
\newcommand{\R}{\mathbb{R}}

\newcommand{\G}{\mathbb{G}}

\newcommand{\sumin}{\sum_{i=1}^n}
\newcommand{\asto}{\overset{\text{a.s.}}{\to}}
\newcommand{\weakto}{\rightsquigarrow}

\newtheorem{exam}{Example}

\newtheorem{lem}{Lemma}

\newtheorem{theorem}{Theorem}

\newtheorem{assumption}{Assumption}

\title{\textbf{Closed-form estimation and inference\\for panels with attrition and refreshment samples}}

\date{\today}

\author[1]{\textsc{Grigory Franguridi}}
\author[2]{\textsc{Lidia Kosenkova}}

\affil[1]{Center for Economic and Social Research, University of Southern California \vspace{1ex}}
\affil[2]{Department of Economics, University of Virginia}

\begin{document}
\maketitle

\begin{abstract}
    \onehalfspacing
    It has long been established that, if a panel dataset suffers from attrition, auxiliary (refreshment) sampling restores full identification under additional assumptions that still allow for nontrivial attrition mechanisms.
    Such identification results rely on implausible assumptions about the attrition process or lead to theoretically and computationally challenging estimation procedures. 
    We propose an alternative identifying assumption that, despite its nonparametric nature, suggests a simple estimation algorithm based on a transformation of the empirical cumulative distribution function of the data.
    This estimation procedure requires neither tuning parameters nor optimization in the first step, i.e., it has a \emph{closed form}.
    We prove that our estimator is consistent and asymptotically normal and demonstrate its good performance in simulations.
    We provide an empirical illustration with income data from the Understanding America Study.

\bigskip


\bigskip


\end{abstract}

\newpage

\onehalfspacing
\frenchspacing

\section{Introduction}\label{sec:intro}

Attrition in panel data is a widespread phenomenon. Units tracked over time may drop out of the sample for several reasons, including self-selection, non-survival, and increasing survey burden.
When attrition is nonrandom, the ensuing bias in structural estimates is difficult to handle both theoretically and computationally.
Complete debiasing is impossible without either strong (parametric or nonparametric) assumptions on the attrition process, such as the selection on observables, or the availability of auxiliary data.
The latter often comes in the form of \emph{refreshment samples}, i.e., extra random samples from the population in the drop-out period, see, e.g., \citet{deng2013handling,taylor2020evaluating,watson2021refreshment}.
Refreshment samples are available in many widely used survey panels such as the RAND Malaysian Family Life Survey (MFLS), the Medical Expenditure Panel Survey (MEPS), and the Current Population Survey (CPS).
In a seminal paper, \citet{hirano2001combining} proved that, with refreshment samples, identification is restored under a quasi-separability assumption on the attrition process, which they call \emph{additive nonignorability}.
Subsequent papers developed estimation and inference procedures relying on assumptions of varying strength.
For example, \citet{bhattacharya2008inference} proposed a sieve-based semiparametric, asymptotically normal estimator under additive nonignorability, while \citet{hoonhout2019nonignorable} derived a two-step GMM procedure for the case of multi-wave panels.
Recently, \citet{franguridi2024estimation} suggested a computationally tractable estimation procedure using iterative proportional fitting (raking), and \citet{franguridi2024robust} provided a debiased version of the raking estimator along with an influence function-based estimator of its asymptotic variance.
Alternative approaches to estimation and inference with refreshment samples or other auxiliary data include \citet{hellerstein1999imposing,nevo2003using,d2010new,si2015semi,sadinle2019sequentially,franguridi2025inference,franguridi2025generalized}, among others.

The aforementioned estimation techniques either require strong assumptions on the attrition process or lead to computationally challenging multistep estimators.
Our main contribution is an estimation procedure that relies on transformations of the empirical cumulative distribution function of the data,  does not require either functional optimization or the choice of tuning parameters, and works for both discrete and continuous data; we call this procedure \emph{closed-form}.
This procedure is particularly attractive in the presence of high-dimensional covariates due to its immunity to the curse of dimensionality and admits bootstrap inference.
These advantages come at the cost of a slightly nonstandard quasi-separability assumption on the attrition process.
We consider it a fair price for the simplicity of theoretical analysis, computational feasibility, and absence of tuning parameters.

The rest of the paper is organized as follows.
\Cref{sec:framework} introduces the model and derives the key identification result.
\Cref{sec:estimation} presents our closed-form estimator, its asymptotic analysis, and the construction of confidence intervals.
\Cref{sec:mc} illustrates the performance of the estimator and the confidence intervals in Monte Carlo simulations.
\Cref{sec:empirical} presents an empirical application.
\Cref{sec:conclusion} concludes.

%Notation: for vectors $x=(x_1,\dots,x_d)$ and $y=(y_1,\dots,y_d)$, we write $x \le y$ if $x_i \le y_i$ for all $i=1,\dots,d$.

\section{Framework and identification}\label{sec:framework}

We follow the setup of \citet{hirano2001combining}.

Let $Z_{it} = (X_{it}',Y_{it})' \in \R^{d}$ denote the stacked vector of covariates and outcomes for unit $i$ at time $t = 1,2$.
We observe $Z_{it}$ for a random sample of units $i=1,\dots,n_1$ at time $t=1$.
There is no initial nonresponse.
However, at time $t=2$, the units may drop out of the sample.
Let $W_i$ be the indicator of unit $i$ staying in the sample and suppose, without loss of generality, that the stayers are units $i=1,\dots,n_2$.
In addition to this incomplete panel, we observe an auxiliary (refreshment) sample $Z_{i2}^r$, $i=1,\dots,n_r$, from the unconditional distribution of $Z_{i2}$.
Let $F$ be the cumulative distribution function (CDF) of $Z = (Z_1,Z_2)$ (where we drop the unit subscript).
We are interested in estimating and conducting inference for a parameter $\theta \in \Theta \subset \R^{d_\theta}$ defined by the moment conditions
\begin{align}
    \E_F m(Z_1,Z_2;\theta) = \int m(z_1,z_2,\theta) \, dF(z_1,z_2) = 0,\label{eq:mom-cond-1}
\end{align}
where $m: \R^{2d} \times \Theta \to \R^{d_m}$ is a known moment function.
We assume for simplicity that $d_m = d_\theta$, but our results can be generalized to an arbitrary number of moments.

This framework is very general and includes the estimands of interest in both linear and nonlinear panel data models.
Although we focus on the classical case of the target parameter defined by the moment conditions \eqref{eq:mom-cond-1}, the results of this paper hold as long as $\theta=\theta(F)$ is a Hadamard differentiable functional of $F$.
We illustrate the broad applicability of our setup with a series of examples.

\begin{exam}[linear regression with two-way fixed effects]
    Consider the standard linear regression with two-way fixed effects,
\begin{align*}
    y_{it} = \alpha_i + f_t + x_{it}'\theta + \varepsilon_{it}, \quad i=1,\dots, n, \,\, t=1,2.
\end{align*}
A huge literature is devoted to the identification of the slope coefficients $\theta$ under various assumptions.
When the covariates are strictly exogenous, one can use the within transformation $\ddot{\zeta}_{it} = \zeta_{it} - \frac{1}{n} \sum_{i=1}^n \zeta_{it} - \frac{1}{2}(\zeta_{i1}+\zeta_{i2}) + \frac{1}{2n} \sum_{i=1}^n (\zeta_{i1}+\zeta_{i2})$ to identify $\theta$ as the OLS coefficient in the transformed regression, 
\begin{align*}
    \E \left( \ddot{y}_{it} - \ddot{x}_{it}'\theta \right) \ddot{x}_{it} = 0.
\end{align*}
This is the moment condition of the form \eqref{eq:mom-cond-1}, and hence can be estimated under attrition when a refreshment sample is available in the second period.

When $x_{it}$ contains lagged outcomes, at least three periods are needed to estimate the above model (for example, via the Arellano-Bond GMM).
Our methodology can be extended to handle more than two periods along the lines of \citet{hoonhout2019nonignorable}.
\end{exam}

\begin{exam}[difference-in-differences]

Consider the classical framework in which outcomes $y_{it}$ are tracked for individuals $i$ over periods $t=1,2$, and some individuals are treated in period $2$, which is denoted by $d_{i2}=1$. The standard diff-in-diff estimand is
\begin{align*}
    DID := \E \left[ y_{i2}-y_{i1} \,\vert\, d_{i2}=1 \right] - \E \left[ y_{i2}-y_{i1} \,\vert\, d_{i2}=0 \right],
\end{align*}
which, under the parallel trends assumption, is equal to the average treatment effect on the treated, ATT. Using the potential outcome notation,
\begin{align*}
    ATT &:= \E[y_{i2}(1) - y_{i2}(0) \,\vert\, d_{i2}=1 ] = DID.
\end{align*}
When there is attrition in the second period, there are two cases depending on whether the treatment status $d_{i2}$ is observable for all units, including those who drop out in the second period.

In the first case, when the treatment status $d_{i2}$ is observed for all units, the model is identified with period-1 and refreshment samples only, without the need for the stayers in period 2 (i.e., with repeated cross-sections only).

In the second case, when the treatment status $d_{i2}$ is only observed for stayers (which is common in panel surveys), $\E[y_{i1}|d_{i2} = 0]$ and $\E[y_{i1}|d_{i2} = 1]$ cannot be calculated using only marginal distributions in two time periods. Instead, the joint distribution of outcome and treatment status over the two periods is required.
To see it formally, rewrite
\[
ATT = \E\left[ \frac{d_{i2}(y_{i2}-y_{i1})}{q_i} - \frac{(1-d_{i2})(y_{i2}-y_{i1})}{1-q_i} \right],
\]
where $q_i=P(d_{i2}=1)$ is the propensity score.
Since ATT contains expressions of the form $d_{i2} y_{i1}$ and $(1-d_{i2})y_{i1}$, it is not additively separable in the two periods, and hence requires panel data for identification.

\end{exam}

\begin{exam}[quantile treatment effects]
Consider the panel data with $T = 3$ periods, where the treatment only occurs in the last period.
The data are a random sample from $(y_1,y_2,y_3,d)$, where $y_t$ is the observed outcome at time $t$ and $d \in \{0,1\}$ is the indicator of treatment.
Suppose that the object of interest is the quantile treatment effect on the treated
\begin{align*}
QTT(\tau) = F_{y_3(1)|d=1}^{-1}(\tau) - F_{y_3(0)|d=1}^{-1}(\tau).
\end{align*}
The first term is identified directly from the data. For the second term, \citet{callaway2019quantile} show that, under their assumptions of \emph{distributional parallel trends} and \emph{copula stability},
\begin{align*}
F_{y_3(0)|d=1}(y) = \Prb & \left[ F_{\Delta y_3|d=0}^{-1}\left( F_{\Delta y_{2}|d=1}(\Delta y_{2}) \right) \le y - F_{y_{2}|d=1}^{-1} \left(F_{y_{1}|d=1}(y_{1}) \right) \,\vert\, d=1\right].
\end{align*}
The distributions $F_{\Delta y_3|d=0}$, $F_{\Delta y_2|d=1}$, $F_{y_2|d=1}$, and $F_{y_1|d=1}$ are identified directly from the data, and hence the second term in the QTT is also identified.
Therefore, we can write the QTT as a (complicated) nonlinear functional of the joint distribution of $(y_1,y_2,y_3,d)$.

Suppose some units may drop out from the sample in period $t=3$, but a refreshment sample is available.
Then the QTT model fits our framework with the first two periods combined into one, i.e., with $z_1=(y_1,y_2)$ and $z_2 = (y_3,d)$.
\end{exam}

Now, we introduce our main identifying assumption for point identification of the joint distribution of $Z_1$ and $Z_2$. Let $F^w$, $F_1$, and $F_2$ be the CDFs of $(Z_1,Z_2)|W=1$, $Z_1$, and $Z_2$, respectively.
These distributions can be readily estimated from the balanced panel (retaining stayers only), the first-period sample, and the refreshment sample, respectively.
The key identity relating the target distribution $F$ to the data is
\begin{align}
F(z_1,z_2) = \frac{\Prb(W=1)}{\Prb(W=1|Z_1\le z_1,Z_2\le z_2)} F^w(z_1,z_2). \label{eq:F-via-Fw}
\end{align}
The weight $\Prb(W=1) / \Prb(W=1|Z_1\le z_1,Z_2\le z_2)$ is not identified without further restrictions.
To see why, notice that this object is an unrestricted function of the joint distribution of $Z_1,Z_2$, while the information available for its identification is only the two marginal CDFs $F_1$ and $F_2$.
To close the gap, we impose a separability assumption on the weight.\footnote{If no assumptions are imposed on the attrition process, any distribution with marginals $F_1$ and $F_2$ is consistent with the data. Hence, in most cases, the partial identification approach will not lead to informative bounds on the structural parameter.}

\begin{assumption}\label{a:identif}
$\Prb(W=1|Z_1\le z_1,Z_2\le z_2) = G(k_1(z_1)+k_2(z_2))$ for a known, differentiable, strictly increasing function $G: \R \to (0,\infty)$ and some unknown functions $k_1: \R^{d} \to \R$, $k_2:\R^{d} \to \R$.\footnote{We believe that choosing the link function $G$
is not a critical decision in practice, as documented, e.g., in \citet{little1991models,hirano2001combining}. However, formally comparing the classes of selection mechanisms captured by \cref{a:identif} when varying the link function is outside of the scope of this paper.}

\end{assumption}

This assumption is compatible with the missing-completely-at-random condition ($k_1=k_2=const$).
It leads to an explicit identification of $k_1$ and $k_2$, see \Cref{thm:identification} below.
Besides, it neither implies nor is implied by the analogous assumption on $\Prb(W=1|Z_1=z,Z_2=z)$ in \citet{hirano2001combining}, viz.
\begin{align}
    \Prb(W=1|Z_1=z_1,Z_2=z_2) = \tilde G(\tilde k_1(z_1) + \tilde k_2(z_2)) \label{eq:HIRR-assum}
\end{align}
for some known link function $\tilde G$ and unrestricted functions $\tilde k_1,\tilde k_2$. The following example illustrates this point.

\begin{exam}
  \textnormal{  Suppose $Z_1,Z_2 \in [0, 1]^2$ with density $f(z_1, z_2) = z_1 + z_2 $ and the conditional probability of staying is given by
    \begin{align*}
        \Prb(W=1|Z_1=z_1,Z_2=z_2) = az_1^2 + b z_1 z_2 + a z_2^2.
    \end{align*}
    for some constants $a,b$.
    Using the formula
    \begin{align*}
        \Prb(W=1|Z_1\le z_1,Z_2\le z_2)& = \frac{1}{F(z_1,z_2)} \int_{-\infty}^{z_1} \int_{-\infty}^{z_2} \Prb(W=1|Z_1=t_1,Z_2=t_2) f(t_1,t_2) \, dt_1 dt_2=\\
        & =\frac{3az_1^3+2(a+b)z_1^2z_2+2(a+b)z_1z_2^2+3az_z^3}{6(z_1+z_2)},
    \end{align*}
    it is easy to show that (a) when $a=2/11$, $b=7/11$, \Cref{a:identif} holds with $G(x)=x^2/11$, but \eqref{eq:HIRR-assum} does not hold; (b) when $a=1/2$, $b=0$, \Cref{a:identif} does not hold, while \eqref{eq:HIRR-assum} holds; (c) when $a=0$, $b=1$, both \Cref{a:identif} and \eqref{eq:HIRR-assum} hold.
    Put differently, even when the alternative assumption \eqref{eq:HIRR-assum} is imposed, our assumption is still valid for a class of DGPs with nontrivial attrition processes.}
\end{exam}

The next example shows that the left-hand side of \eqref{eq:F-via-Fw} can be a valid CDF under \Cref{a:identif}.

\begin{exam}
\textnormal{    Suppose $Z_1,Z_2$ are scalar random variables and $G(x)=\exp(x)$. We have
    \[
    F(z_1,z_2) \propto \frac{F^w(z_1,z_2)}{\exp(k_1(z_1)+k_2(z_2))},
    \]
    or, taking the mixed partial derivative,
    \[
    f(z_1,z_2) \propto \frac{(1-k_1'(z_1)D_{z_1}F^w )(1- k_2'(z_2)D_{z_2}F^w ) + (f^w-1)}{\exp(k_1(z_1)+k_2(z_2))}.
    \]
    This is a proper density if the numerator is nonnegative. Since $D_{z_1}F^w$, $D_{z_2} F^w$, and $f^w$ are nonnegative functions, a simple sufficient condition is for \emph{$k_1$ and $k_2$ to be decreasing}.}
    
  \textnormal{  For a more concrete example, consider $Z_1,Z_2 |W=1 \sim \text{ iid U}[0,1]$ and $k_1(z_1)=a + c_1 z_1$, $k_2(z_2)=b+c_2z_2$. Then the aforementioned condition holds whenever $c_1,c_2 \le 1$. Normalization is achieved by an appropriate choice of $a$ and $b$.}
\end{exam}

\Cref{a:identif} leads to a \emph{closed-form} identification result for the unknown function $k_1(z_1)+k_2(z_2)$ which informs an estimation procedure based on empirical CDFs.
This is in contrast to \citet{hirano2001combining} that identifies $\tilde k_1(z_1)+\tilde k_2(z_2)$ implicitly as a solution to a nonlinear functional optimization problem which makes the estimation procedure hard to analyze and implement, see \citet{franguridi2024robust}.

\begin{theorem}\label{thm:identification}
Under \Cref{a:identif}, the target CDF can be written as
\begin{align}
F(z_1,z_2) = \Phi\left(\Prb(W=1),F_1(z_1),F_2(z_2),F_1^w(z_1),F_2^w(z_2),F^w(z_1,z_2) \right), \label{eq:F-via-Phi}
\end{align}
where the function $\Phi$ is defined by
\begin{align}
    \Phi(p,F_1,F_2,F_1^w,F_2^w,F^w) = \frac{p F^w}{G\left(G^{-1}\left(\frac{p F_1^w}{F_1} \right) + G^{-1}\left(\frac{p F_2^w}{F_2} \right) - G^{-1}\left(p\right) \right)}.
\end{align}
\end{theorem}

\begin{proof}
Denote $p=\Prb(W=1)$.
Equation \eqref{eq:F-via-Fw} implies
\begin{align*}
G(k_1(z_1)+k_2(z_2)) = \frac{pF(z_1,z_2|W=1)}{F(z_1,z_2)}.
\end{align*}
Since $G$ is strictly increasing, the inverse function $G^{-1}$ exists.
Plugging in $z_1=\infty$ and/or $z_2=\infty$, we get
\begin{align*}
k_1(z_1) + k_2(\infty) &= G^{-1}\left(\frac{p F_1(z_1|W=1)}{F_1(z_1)} \right), \\
k_1(\infty) + k_2(z_2) &= G^{-1}\left(\frac{p F_2(z_2|W=1)}{F_2(z_2)} \right), \\
k_1(\infty)+k_2(\infty) &= G^{-1}\left(p\right).
\end{align*}
Normalizing (say) $k_2(\infty)=0$, we obtain, for all $z_1,z_2$,
\begin{align*}
k_1(z_1) &= G^{-1}\left(\frac{p F_1(z_1|W=1)}{F_1(z_1)} \right), \\
k_2(z_2) &= G^{-1}\left(\frac{p F_2(z_2|W=1)}{F_2(z_2)} \right) - G^{-1}\left(p\right).
\end{align*}
Substituting these expressions into \eqref{eq:F-via-Fw} completes the proof.
\end{proof}

\section{Estimation and inference}\label{sec:estimation}

\Cref{thm:identification} suggests a two-step estimation procedure.
In the first step, we estimate $F$ by plugging in the empirical CDF estimators of $F_1,F_2,F_1^w,F_2^w,F^w$.
In the second step, we substitute this estimator into the moment conditions \eqref{eq:mom-cond-1} and compute $\hat\theta$ that sets these moment conditions to zero.

To describe this procedure formally, define
\begin{align}
&\hat p = \frac{n_2}{n_1}, \quad
\hat F_1(z_1) = \frac{1}{n_1} \sum_{i \in n_1} 1(z_{1i} \le z_1), \quad
\hat F_1^w(z_1) = \frac{1}{n_2} \sum_{i \in n_2} 1(z_{1i} \le z_1), \\
&\hat F_2(z_2) = \frac{1}{n_r} \sum_{i \in n_r} 1(z_{2i}^r \le z_2), \quad
\hat F_2^w(z_2) = \frac{1}{n_2} \sum_{i \in n_2} 1(z_{2i} \le z_2), \\
&\hat F^w(z_1,z_2) = \frac{1}{n_2} \sum_{i \in n_2} 1(z_{1i} \le z_1, z_{2i} \le z_2).
\end{align}
Then the plug-in estimator of $F$ is
\begin{align*}
    \hat F(z_1,z_2) = \Phi\left( \hat p, \hat F_1(z_1),\hat F_2(z_2),\hat F_1^w(z_1),\hat F_2^w(z_2), \hat F^w(z_1,z_2) \right).
\end{align*}
We now use this estimator to calculate the sample analog of the moment condition \eqref{eq:mom-cond-1} as the integral with respect to the distribution induced by $\hat F$.
Because its arguments are empirical CDFs, $\hat F$ is a piecewise constant function.
Let $\hat{\mathcal{Z}} \subset \R^{2d}$ be the finite set of its discontinuity points.
In our case,
\[
\hat{\mathcal{Z}} := \hat{\mathcal{Z}}_1 \times (\hat{\mathcal{Z}}_2 \cup \hat{\mathcal{Z}}_2^r),
\]
where $\hat{\mathcal{Z}}_1=\{z_{1i}, \, i\in n_1\}$, $\hat{\mathcal{Z}}_2=\{z_{2i}, \, i\in n_2\}$, and $\hat{\mathcal{Z}}_2^r=\{z_{2i}^r, \, i\in n_r\}$ are the respective empirical supports.
Then the induced ``probability'' (``jump size'') of $\hat F$ at any $\zeta=(\zeta_1,\dots,\zeta_{2d})\in\hat{\mathcal{Z}}$ is\footnote{If $\hat F$ were the empirical CDF corresponding to $n$ distinct points $\hat{\mathcal{Z}}=\{\zeta_1,\dots,\zeta_n\}$, then the formula would yield $\hat f(x_i)=1/n$ for all $i=1,\dots,n$.}
\begin{align}
    \hat f(\zeta) = \sum_{(i_1,\dots,i_{2d}) \in \{0,1\}^{2d}} (-1)^{i_1+\dots+i_{2d}} \hat F\left(\zeta_1+(-1)^{i_1} h_1, \dots, \zeta_{2d}+(-1)^{i_{2d}} h_{2d} \right), \label{eq:f-jump}
\end{align}
where $h_k$ is any positive scalar that is less than the minimal gap between the observations of the $k$-th component of $\zeta$, $0<h_k < \min_{i,j:\, \zeta_{k,i} \neq \zeta_{k,j}}|\zeta_{k,i}-\zeta_{k,j}|$.\footnote{The definition of $\hat f$ does not depend on the choice of $h_k$, and so $h_k$ is not a tuning parameter.}
This is a discrete version of the formula
\begin{align*}
    f(\zeta) = \frac{\partial^{2d}}{\partial \zeta_{1} \cdots \partial \zeta_{2d}} F(\zeta_1,\dots,\zeta_{2d})
\end{align*}
which links a CDF $F$ to its PDF $f$ in case the former is differentiable.
We emphasize, however, that the formula \eqref{eq:f-jump} works regardless of whether the data are discrete or continuous.
Notice that, because $\hat F$ does not have to be a valid empirical CDF, $\hat f$ may take negative values in finite samples.\footnote{Rearrangement of $\hat{F}(z_1,z_2)$ similar to \citet{chernozhukov2009improving} does not guarantee that the rearranged function is a proper CDF.}
This does not affect the asymptotic properties of the resulting estimator.

The complete estimation procedure is as follows.

\textbf{Algorithm}.
\begin{enumerate}
    \item Calculate $\hat p$ and the empirical CDFs $\hat F_1, \hat F_2, \hat F_1^w, \hat F_2^w, \hat F^w$;
    \item Plug in to obtain $\hat F = \Phi(\hat p, \hat F_1, \hat F_2, \hat F_1^w, \hat F_2^w, \hat F^w)$;
    \item Calculate the jump size $\hat f(z_{1},z_{2})$ at points $(z_{1},z_{2}) \in \hat{\mathcal{Z}}$ according to formula \eqref{eq:f-jump};
    \item Set $\hat\theta$ such that\footnote{Another way to implement the second step is to draw a random sample from the discrete distribution defined by (a trimmed version of) $\hat f$ and estimate $\theta$ by the conventional GMM based on this sample.}
    \begin{align}
     \int m(z_1,z_2;\hat\theta) \, d\hat F(z_1,z_2) =  \sum_{(z_1,z_2) \in \hat{\mathcal{Z}}} m(z_{1},z_{2};\hat\theta) \hat f(z_{1},z_{2})=0. \label{estimator}
    \end{align}
\end{enumerate}

This procedure has several advantages.
First, it does not require choosing any tuning parameters despite relying on a semiparametric \Cref{a:identif}.
Second, the first-step estimator has an explicit formula and only depends on empirical CDFs, which makes it suitable for high-dimensional data.
Finally, because the first-step estimator is a smooth transformation of empirical CDFs, which jointly converge to a known Gaussian process by Donsker's theorem, this procedure allows for the delta method bootstrap, as the following theorem shows.

%Lemma \eqref{lm:SULLN} gives conditions for the strong uniform law of large numbers for the moment condition on $\theta_0$. Thus, we can apply any standard theorem to show the consistency of the estimator $\hat{\theta}$.
To prove consistency of $\hat\theta$, we follow \citet{newey1994large}.
First, under two different sets of assumptions on $m$, we show convergence in probability of $\int m(z;\theta)\, d\hat{F}(z)$ to $\int m(z;\theta)\, dF(z)$ uniformly over $\theta$.
This is an analog of the uniform law of large numbers (ULLN) with sample averaging replaced by integration with respect to an estimated CDF. 

% Since $\hat{F}$ is not an empirical CDF and we cannot apply the ULLN directly, we need to impose boundedness of the variation of the moment function.
\begin{assumption}\label{a:theta}
\text{}
% The parameter space $\Theta$, the support of the data satisfy the following:
    \begin{enumerate}
        \item The parameter $\theta$ belongs to a compact set $\Theta \subset \R^{d_{\theta}}$;
        \item The true value $\theta$ is a unique solution of \eqref{eq:mom-cond-1} in $\Theta$.
    \end{enumerate}
\end{assumption}

\begin{assumption}\label{a:bounded-support}
\text{}
    \item The true distribution $F$ has bounded support.
\end{assumption}

\begin{assumption}\label{a:Lipschitz}
The moment function satisfies the following:
\begin{enumerate}
    \item $m(Z;\theta)$ is a.s. continuous in $\theta\in\Theta$
    \item There exists $M_1>0$ such that $\Vert m(\cdot ;\theta)\Vert_{Lip}\le M_1$ for all $\theta\in\Theta$, where 
    $$\Vert m(\cdot;\theta)\Vert_{Lip}\equiv \sup_{Z\ne\tilde{Z}} \frac{ \Vert m(Z;\theta)-m(\tilde{Z};\theta)\Vert_{1}}{\Vert Z-\tilde{Z}\Vert_{1}}.$$
\end{enumerate}
\end{assumption}
\Cref{a:Lipschitz} imposes smoothness on the moment function $m(Z;\theta)$, viz. $m(Z;\theta)$ is Lipschitz continuous w.r.t. $Z$ for any fixed $\theta\in\Theta$ with a uniformly bounded Lipschitz constant.

\begin{lem}\label{lem:ULLN_Lipschitz}
Suppose \Cref{a:identif,a:theta,a:Lipschitz,a:bounded-support} hold. Then $M(\theta) = \int m(z;\theta) \, dF(z) $ is continuous and
\begin{equation}
    \sup_{\theta\in\Theta}\left\| \int m(z;\theta)d\hat{F}(z) - \int m(z;\theta)dF(z)\right\|\asto 0. \label{eq:ULLN}
\end{equation}

\end{lem}

Proof can be found in \Cref{proof_lemma_Lipschitz}.

To allow for non-smooth moment functions, such as $m(z;\theta)=1(z\le\theta)$, one can apply \Cref{lem:ULLN-prime}.

\begin{assumption}\label{a:m-Lipschitz}
The moment function satisfies the following:
\begin{enumerate}
    \item\label{m-bv} $m(z;\theta)$ is of bounded variation as a function of $z$ for each $\theta\in\Theta$;
    \item\label{m-cont-in-theta}
    For every $\theta\in\Theta$, the set $\{z\,\vert\,\lim_{\gamma\to\theta}m(z;\gamma) = m(z;\theta)\}$ has probability 1 w.r.t. $F$;
    \item\label{m-dominated} There exists a function $d(z)$ such that $\Vert m(z;\theta) \Vert \le d(z)$ for all $\theta\in\Theta$ and $\int d(z)\, dF(z)<\infty;$
    \item\label{u-bv} Function $u(z;\theta, d) = \sup_{\vert\gamma-\theta\vert\le d}\vert m(z;\gamma)-m(z;\theta)\vert$ is of bounded variation as a function of $z$ for each $\theta\in\Theta$ and for each $d\le \varepsilon$ for some $\varepsilon>0.$
\end{enumerate}
\end{assumption}

\Cref{a:m-Lipschitz} is similar to the classical statement of the uniform law of large numbers, see Lemma 2.4 in \citet{newey1994large}, with the additional high-level requirement on the function $u(z;\theta, d)$ to be of bounded variation.

\begin{lem}\label{lem:ULLN-prime}
    Under \Cref{a:identif,a:theta,a:m-Lipschitz}, the conclusions of \Cref{lem:ULLN_Lipschitz} hold.
\end{lem}

 Proof can be found in \Cref{proof_lem_prime}.
 It follows from Lemma 1 of \citet{tauchen1985diagnostic} with an adjustment to the fact that $\hat{F}$ is not an empirical CDF. 

\begin{theorem}[Consistency]\label{thm:consistency}
If the assumptions of \Cref{lem:ULLN_Lipschitz} or \Cref{lem:ULLN-prime} hold,
then $$\hat{\theta}\asto \theta.$$
\end{theorem}

\Cref{thm:consistency} follows from \Cref{lem:ULLN_Lipschitz}, \Cref{lem:ULLN-prime}, and an argument similar to the one in Theorem 2.1 of \citet{newey1994large}.

Finally, we establish $\sqrt{n}$-consistency and asymptotic normality of $\hat\theta$ and the validity of the nonparametric bootstrap for inference on $\theta$.
To this end, we impose the following assumption that ensures the Hadamard differentiability of $\theta$ as a functional of the distribution of the data.
\begin{assumption}\label{as:M-is-differentiable}
\text{}
    \begin{enumerate}
    \item \label{as:M-diff-in-theta} There exists $\varepsilon>0$ and a neighborhood $\tilde \Theta$ of $\theta$ such that
    \[
    M(\tilde \theta,\tilde F) := \int m(z;\tilde \theta) \, d \tilde F(z)
    \]
    is continuously differentiable in $\tilde \theta \in \tilde \Theta$ for all $\tilde F$ such that $\sup_{z \in \R^{2d}}|\tilde F(z)-F(z)| < \varepsilon$.
    \item\label{as:M-diff-nondegenerate} The Jacobian $\frac{\partial}{\partial \theta} M(\theta,F)$ exists and is nonsingular.
\end{enumerate}    
\end{assumption}
\Cref{as:M-is-differentiable}.\ref{as:M-diff-in-theta} holds if $m(z;\tilde \theta)$ is continuously differentiable in $\tilde \theta$ for $F$-a.e. $z$, without further restrictions on the distribution $F$.
It also holds in the ``minimum distance'' case $m(z;\tilde \theta)=\tilde m(z)-\tilde \theta$.
\Cref{as:M-is-differentiable}.\ref{as:M-diff-nondegenerate} is standard and posits local identification.

\begin{theorem}[Inference]\label{thm:inference}
    Suppose that \Cref{as:M-is-differentiable} holds.
    Then $\sqrt{n}(\hat\theta - \theta)$ converges weakly to a normal distribution that can be consistently estimated by the nonparametric bootstrap.
\end{theorem}

Proof can be found in \Cref{proof_inference}.

\section{Monte Carlo simulation}\label{sec:mc}

\begin{table}[h!]
\centering
\small
\begin{tabular}{c c c c c c }
& & \multicolumn{2}{c}{$n_1=n_r=1000$} & \multicolumn{2}{c}{$n_1=n_r=10,000$}\\											
\cmidrule(lr){3-4} \cmidrule(lr){5-6}											
	&		&	$\hat\theta$	&	$\hat\theta_{naive}$	&	$\hat\theta$	&	$\hat\theta_{naive}$	\\
    \hline
\multirow{6}{*}{$m=5$}	&	bias	&	-0.004	&	-0.208	&	-0.001	&	-0.205	\\
	&	rmse	&	0.229	&	0.279	&	0.071	&	0.213	\\
	&	mae	&	0.183	&	0.236	&	0.057	&	0.205	\\
	&	coverage $99\%$	&	0.993	&		&	0.992	&		\\
	&	coverage $95\%$	&	0.953	&		&	0.953	&		\\
	&	coverage $90\%$	&	0.905	&		&	0.905	&		\\
    \hline
\multirow{6}{*}{$m=10$}	&	bias	&	0.003	&	-0.171	&	0.000	&	-0.171	\\
	&	rmse	&	0.494	&	0.428	&	0.150	&	0.210	\\
	&	mae	&	0.387	&	0.351	&	0.120	&	0.181	\\
	&	coverage $99\%$	&	0.998	&		&	0.993	&		\\
	&	coverage $95\%$	&	0.964	&		&	0.957	&		\\
	&	coverage $90\%$	&	0.907	&		&	0.909	&		\\
    \hline
\multirow{6}{*}{$m=20$}	&	bias	&	-0.005	&	-0.030	&	0.003	&	-0.027	\\
	&	rmse	&	0.881	&	0.836	&	0.271	&	0.264	\\
	&	mae	&	0.688	&	0.680	&	0.215	&	0.212	\\
	&	coverage $99\%$	&	0.998	&		&	0.990	&		\\
	&	coverage $95\%$	&	0.981	&		&	0.949	&		\\
	&	coverage $90\%$	&	0.950	&		&	0.902	&		
\end{tabular}
\caption{Simulation results for the discrete DGP. Number of simulations $S=5000$. Bias, RMSE, and MAE are reported as shares of the true value of $\theta(m)$.}
\label{tab-mc-discr}
\end{table}

\begin{table}
\centering
\small
\begin{tabular}{c c c c c c }
& & \multicolumn{2}{c}{$n_1=n_r=1000$} & \multicolumn{2}{c}{$n_1=n_r=5000$}\\									
\cmidrule(lr){3-4} \cmidrule(lr){5-6}	
&		&	$\hat\theta$	&	$\hat\theta_{naive}	$ &	$\hat\theta$	&	$\hat\theta_{naive}	$\\
\hline
\multirow{6}{*}{$\nu=0.2$}	&	bias	&	0.001	&	0.001	&	0.000	&	0.001	\\
	&	rmse	&	0.009	&	0.014	&	0.004	&	0.007	\\
	&	mae	&	0.007	&	0.012	&	0.003	&	0.005	\\
	&	coverage $99\%$	&	0.984	&	&	0.992	&	\\
	&	coverage $95\%$	&	0.942	&	&	0.964	&	\\
	&	coverage $90\%$	&	0.878	&	&	0.922	&	\\
\hline
\multirow{6}{*}{$\nu=0.5$}	&	bias	&	0.001	&	0.001	&	0.000	&	0.002	\\
	&	rmse	&	0.010	&	0.016	&	0.004	&	0.007	\\
	&	mae	&	0.008	&	0.013	&	0.003	&	0.006	\\
	&	coverage $99\%$	&	0.980	&	&	0.988	&		\\
	&	coverage $95\%$	&	0.936	&	&	0.96	&	\\
	&	coverage $90\%$	&	0.892	&	&	0.926	&	\\
\hline
\multirow{6}{*}{$\nu=0.8$}	&	bias	&	0.001	&	0.001	&	0.001	&	0.002	\\
	&	rmse	&	0.010	&	0.017	&	0.004	&	0.008	\\
	&	mae	&	0.008	&	0.014	&	0.003	&	0.006	\\
	&	coverage $99\%$	&	0.980	&	&	0.976	&	\\
	&	coverage $95\%$	&	0.934	&	&	0.942	&	\\
	&	coverage $90\%$	&	0.886	&	&	0.908	&	\\
\end{tabular}
\caption{Simulation results for the continuous DGP. Number of simulations $S=500$.}
\label{tab-mc-cont}
\end{table}

In this section, we illustrate the performance of our two-step estimator and the associated bootstrap confidence intervals in simulations.
We abstract away from both conditioning on the covariates and nonlinear moments, as in the examples above, and employ two simple data-generating processes (DGPs), one with discrete variables and one with continuous variables.
Both DGPs employ the logit link function $G(x) = (1+e^{-x})^{-1}$ for the conditional selection probability.\footnote{We explore the sensitivity of our methodology to misspecification of the link function in \Cref{sec:misspec}.}

\textbf{Discrete data}

Our first DGP is a discrete Markov process with scalar outcomes.
We posit that $Z_1$ has the uniform distribution over $\{1,\dots,m\}$, where $m \in \{5,10,20\}$, and $Z_2$ is determined from $Z_1$ given a transition matrix with positive elements.
The attrition functions are 
\begin{align*}
    k_1(z_1)=c_1 z_1, \quad k_2(z_2) = c_2 z_2,
\end{align*}
where the constants $c_1,c_2$ and the transition matrix are chosen so that the unconditional attrition rate is $30\%$.
The target parameter is $\theta(m) = \Prb_m(Z_2=1|Z_1=1)$ with true values
\begin{align*}
    \theta(5) = 0.18, \quad \theta(10) = 0.09, \quad \theta(20) = 0.04.
\end{align*}
Since the parameter of interest is a probability that decreases with the number of support points, we present the bias, root mean-squared error (RMSE), and mean absolute deviation (MAE) as shares of the true value of $\theta(m)$.

\textbf{Continuous data}

Our second DGP is a copula model for continuous variables $Z_1$ and $Z_2$.
The (target) joint distribution is modeled as a Gaussian copula
\begin{align*}
    F(z_1,z_2) = \Phi_\nu(\Phi^{-1}(z_1),\Phi^{-1}(z_2)), \quad (z_1,z_2) \in [0,1]^2,
\end{align*}
where $\Phi_\nu$ is a bivariate Gaussian CDF with zero means, unit variances, and correlation $\nu \in \{0.2,0.5,0.8\}$, and $\Phi$ is the CDF of $N(0,1)$.
The functions in the attrition mechanism are
\begin{align*}
    k_1(z_1)=c_0+c_1z_1, \quad k_2(z_2)=c_2 z_2,
\end{align*}
where the constants $c_0,c_1,c_2$ are chosen so that the overall attrition rate is about $70\%$.
The target parameter is $\theta(\nu)=\E_{\nu}[Z_1 Z_2]$ with the true values
\begin{align*}
    \theta(0.2) \approx 0.27, \quad \theta(0.5)\approx 0.29, \quad \theta(0.8)\approx 0.32.
\end{align*}

Tables \ref{tab-mc-discr} and \ref{tab-mc-cont} present the simulation results.
We report bias, root mean squared error (RMSE), and mean absolute error (MAE) of our estimator and of the ``naive'' estimator that ignores attrition and only uses the balanced panel.
We also calculate empirical coverages of the bootstrap confidence interval for nominal confidence levels $1-\alpha \in \{0.90,0.95,0.99\}$. We do not report the empirical coverages for the naive estimator due to its inconsistency.

Unsurprisingly, the naive estimator exhibits larger biases across various specifications and sample sizes, while our (consistent) estimator performs well, with both the bias and the RMSE decreasing with the sample size. The confidence interval coverage is close to nominal for both the discrete DGP and the continuous DGP.

\section{Empirical illustration}\label{sec:empirical}

We illustrate our methodology in a simple model of the change of income over the life course. We use survey data from the Understanding America Study (UAS) conducted by the \citet{uas2018}, a large household panel collected and maintained by the USC Center for Economic and Social Research. We estimate the following model,
\begin{align}
    \text{income}_{it} = \alpha_i + \lambda_t + \theta_1 \text{age}_{it} + \theta_2 \text{age}_{it}^2 + \varepsilon_{it}, \quad t=1,2,
\end{align}
where $\text{income}_{it}$ is the inverse hyperbolic sine of the household income.\footnote{This transformation approximates the logarithm of income for sufficiently large values of income while retaining zero income values.}

We use UAS wave 14 (year 2018) as the first period and UAS wave 15 (year 2020) as the second period.\footnote{Available at https://uasdata.usc.edu/page/Comprehensive+File+And+Panel+Dataset}
The number of households responding in the first period is $8145$, out of which $6432$ households also respond in the second period, with an attrition rate $21\%$.
There are 3983 households sampled in the second period that are not part of the sample in the first period.
These households serve as the refreshment sample.
We set the link function to be logistic, $G(x) = 1/(1+e^{-x})$.

\Cref{tab:emp-app} contains the estimates and the bootstrap standard errors for our procedure using the refreshment sample and for the naive procedure using only the balanced panel. Neither of the estimates is significant at the 5\% nominal level. However, the effect size is substantial: at the average level of income of the working population in 2018 (which is \$60,465), an additional year of life is associated with an average increase of $\sinh(\operatorname{asinh}(\$60,465)+\hat\theta_1+\hat\theta_2)-\$60,465 =\$24,824$ in annual household income.

\begin{table}[h!]
    \centering
    \begin{tabular}{c|c|c|c|c}
       & $\hat\theta_1$ & $\hat\theta_2$ & $\hat\theta_{1,\text{naive}}$ & $\hat\theta_{2,\text{naive}}$  \\
       \hline
       \text{estimate} & 0.347 & -0.003 & 0.165 & 0.000 \\
        \text{st.err.}  & (0.209) & (0.002) & (0.099) & (0.001) \\
    \end{tabular}
    \caption{Estimation results for the static linear regression.}
    \label{tab:emp-app}
\end{table}

\section{Conclusion}\label{sec:conclusion}

In most panel datasets encountered in practice, some units are only observed up to a non-terminal period. Such attrition may lead to significant bias in the estimates and invalid inference. Fortunately, the availability of refreshment samples allows one to remove the bias and restore valid inference under reasonably weak assumptions on the attrition process.

In this paper, we introduce one such assumption that leads to simple, computationally feasible estimators that admit bootstrap inference.
We hope this paper may serve as the first step towards developing feasible estimation and inference algorithms for data structures with nonrandom missingness and auxiliary information.

Topics for further research include an extension to multiple periods, more general missingness patterns, or high-dimensional covariates; dealing with initial non-response; theoretical and computational analysis of behavior under misspecification; establishing a semiparametric efficiency bound; and deriving optimal estimators.

\section*{Acknowledgements}

The project described in this publication relies on data from surveys administered by the Understanding America Study (UAS), which is maintained by the Center for Economic and Social Research (CESR) at the University of Southern California. The project was supported by the National Institute on Aging of the National Institutes of Health and, in part, by the Social Security Administration under Award Number U01AG077280. The content is solely the responsibility of the authors and does not necessarily represent the official views of the National Institutes of Health, USC, or UAS.

We appreciate valuable comments and suggestions from Isaiah Andrews, Tim Armstrong, Arie Kapteyn, Sergey Lototsky, Kirill Ponomarev, John Pepper, Geert Ridder, Fedor Sandomirskiy, Alexander Shapoval, seminar participants at Princeton University and University of Rochester, and conference participants at 16th Greater New York Metropolitan Area Econometrics Colloquium, 2024 North American Summer Meeting of the Econometric Society, Western Economic Association International 99th Annual Conference, and Econometric Society Summer School in Dynamic Structural Econometrics.
All errors and omissions are our own.

\bibliographystyle{apalike}
\bibliography{references}

@article{hirano2001combining,
  title={Combining Panel Data Sets with Attrition and Refreshment Samples},
  author={Hirano, Keisuke and Imbens, Guido W and Ridder, Geert and Rubin, Donald B},
  journal={Econometrica},
  volume={69},
  number={6},
  pages={1645--1659},
  year={2001},
  publisher={Wiley Online Library}
}

@article{hoonhout2019nonignorable,
  title={Nonignorable attrition in multi-period panels with refreshment samples},
  author={Hoonhout, Pierre and Ridder, Geert},
  journal={Journal of Business and Economic Statistics},
  volume={37},
  number={3},
  pages={377--390},
  year={2019},
  publisher={Taylor \& Francis}
}

@article{nevo2003using,
  title={Using weights to adjust for sample selection when auxiliary information is available},
  author={Nevo, Aviv},
  journal={Journal of Business and Economic Statistics},
  volume={21},
  number={1},
  pages={43--52},
  year={2003},
  publisher={Taylor \& Francis}
}

@article{deng2013handling,
  title={Handling Attrition in Longitudinal Studies: The Case for Refreshment Samples},
  author={Deng, Yiting and Hillygus, D Sunshine and Reiter, Jerome P and Si, Yajuan and Zheng, Siyu},
  journal={Statistical Science},
  volume={28},
  number={2},
  pages={238--256},
  year={2013}
}

@article{taylor2020evaluating,
  title={Evaluating supplemental samples in longitudinal research: Replacement and refreshment approaches},
  author={Taylor, Laura K and Tong, Xin and Maxwell, Scott E},
  journal={Multivariate Behavioral Research},
  volume={55},
  number={2},
  pages={277--299},
  year={2020},
  publisher={Taylor \& Francis}
}

@article{watson2021refreshment,
  title={Refreshment sampling for longitudinal surveys},
  author={Watson, Nicole and Lynn, Peter},
  journal={Advances in longitudinal survey methodology},
  pages={1--25},
  year={2021},
  publisher={Wiley Online Library}
}

@article{si2015semi,
  title={Semi-parametric selection models for potentially non-ignorable attrition in panel studies with refreshment samples},
  author={Si, Yajuan and Reiter, Jerome P and Hillygus, D Sunshine},
  journal={Political Analysis},
  volume={23},
  number={1},
  pages={92--112},
  year={2015},
  publisher={Cambridge University Press}
}

@book{villani2009optimal,
  title={Optimal transport: old and new},
  author={Villani, C{\'e}dric and others},
  volume={338},
  year={2009},
  publisher={Springer}
}

@article{young1917multiple,
  title={On multiple integration by parts and the second theorem of the mean},
  author={Young, WH},
  journal={Proceedings of the London Mathematical Society},
  volume={2},
  number={1},
  pages={273--293},
  year={1917},
  publisher={Oxford University Press}
}

@book{vaart2023empirical,
  title={Weak Convergence and Empirical Processes: With Applications to Statistics},
  author={van der Vaart, A. and Wellner, Jon},
  year={2023},
  publisher={Springer}
}

@article{bhattacharya2008inference,
  title={Inference in panel data models under attrition caused by unobservables},
  author={Bhattacharya, Debopam},
  journal={Journal of Econometrics},
  volume={144},
  number={2},
  pages={430--446},
  year={2008},
  publisher={Elsevier}
}

@article{d2010new,
  title={A new instrumental method for dealing with endogenous selection},
  author={d’Haultfoeuille, Xavier},
  journal={Journal of Econometrics},
  volume={154},
  number={1},
  pages={1--15},
  year={2010},
  publisher={Elsevier}
}

@article{hellerstein1999imposing,
  title={Imposing moment restrictions from auxiliary data by weighting},
  author={Hellerstein, Judith K and Imbens, Guido W},
  journal={Review of Economics and Statistics},
  volume={81},
  number={1},
  pages={1--14},
  year={1999},
  publisher={MIT Press 238 Main St., Suite 500, Cambridge, MA 02142-1046, USA journals~…}
}

@article{sadinle2019sequentially,
  title={Sequentially additive nonignorable missing data modelling using auxiliary marginal information},
  author={Sadinle, Mauricio and Reiter, Jerome P},
  journal={Biometrika},
  volume={106},
  number={4},
  pages={889--911},
  year={2019},
  publisher={Oxford University Press}
}

@book{newey1994large,
  title={Large sample estimation and hypothesis testing},
  author={Newey, Whitney K and McFadden, Daniel},
  journal={Handbook of Econometrics,},
  volume={ 4},
  pages={2111--2245},
  year={1994},
  publisher={Amsterdam:
North-Holland.}
}

@article{tauchen1985diagnostic,
  title={Diagnostic testing and evaluation of maximum likelihood models},
  author={Tauchen, George},
  journal={Journal of Econometrics},
  volume={30},
  number={1-2},
  pages={415--443},
  year={1985},
  publisher={Elsevier}
}

@article{franguridi2024estimation,
  title={Raking for estimation and inference in panel models with nonignorable attrition and refreshment},
  author={Franguridi, Grigory and Hahn, Jinyong and Hoonhout, Pierre and Kapteyn, Arie and Ridder, Geert},
  journal={arXiv preprint arXiv:2512.13270},
  year={2025}
}

@article{franguridi2025inference,
  title={Inference in partially identified moment models via regularized optimal transport},
  author={Franguridi, Grigory and Liu, Laura},
  journal={arXiv preprint arXiv:2512.18084},
  year={2025}
}

@article{franguridi2025generalized,
  title={Generalized method of moments with partially missing data},
  author={Franguridi, Grigory and Moon, Hyungsik Roger},
  journal={arXiv preprint arXiv:2511.21988},
  year={2025}
}

@article{franguridi2024robust,
  title={Robust estimation and inference with given marginals},
  author={Franguridi, Grigory and Hahn, Jinyong and Ridder, Geert},
  journal={Working paper, Center for Economic and Social Research, University of Southern California},
  year={2026}
}

@article{callaway2019quantile,
  title={Quantile treatment effects in difference in differences models with panel data},
  author={Callaway, Brantly and Li, Tong},
  journal={Quantitative Economics},
  volume={10},
  number={4},
  pages={1579--1618},
  year={2019},
  publisher={Wiley Online Library}
}

@book{lang2012fundamentals,
  title={Fundamentals of differential geometry},
  author={Lang, Serge},
  volume={191},
  year={2012},
  publisher={Springer Science \& Business Media}
}

@article{chernozhukov2009improving,
  title={Improving point and interval estimators of monotone functions by rearrangement},
  author={Chernozhukov, Victor and Fernandez-Val, Ivan and Galichon, Alfred},
  journal={Biometrika},
  volume={96},
  number={3},
  pages={559--575},
  year={2009},
  publisher={Oxford University Press}
}

@article{little1991models,
  title={Models for contingency tables with known margins when target and sampled populations differ},
  author={Little, Roderick JA and Wu, Mei-Miau},
  journal={Journal of the American Statistical Association},
  volume={86},
  number={413},
  pages={87--95},
  year={1991},
  publisher={Taylor \& Francis}
}

@misc{uas2018,
  author       = {{University of Southern California}},
  title        = {Understanding {A}merica {S}tudy},
  year         = {2018},
  howpublished = {\url{https://uasdata.usc.edu/page/Comprehensive+File+And+Panel+Dataset}},
  note         = {Produced by the USC Dornsife Center for Economic and Social Research, with funding from the National Institute on Aging and the Social Security Administration. Accessed  2018}
}

\newpage

\appendix
\appendixpage

\begin{subappendices}

\section{Proof of \Cref{lem:ULLN_Lipschitz}}\label{proof_lemma_Lipschitz}

\Cref{a:m-Lipschitz}.\ref{m-cont-in-theta} implies by dominated convergence that $M(\theta)$ is a continuous function of $\theta$. 

The multivariate Glivenko-Cantelli theorem implies that, with probability 1,
\begin{align*}
    &\hat F_1 \to F_1, \quad \hat F_2 \to F_2, \\
    &\hat F_1^w \to F_1^w, \quad \hat F_2^w \to F_2^w, \\
    &\hat F^w \to F^w, \quad \hat p \to p,
\end{align*}
where the underlying functional spaces are equipped with uniform norms.
Since $\Phi$ is continuous, by the continuous mapping theorem, with probability 1,
\begin{align}
    \hat F &= \Phi(\hat p,\hat F_1,\hat F_2,\hat F_1^w, \hat F_2^w, \hat F^w) \to \Phi(p, F_1, F_2, F_1^w, F_2^w, F^w) = F. \label{eq:F-hat-to-F}
\end{align}
Let $\hat\mu$ be the signed measure corresponding to the step function $\hat F$, and let $\mu$ be the probability measure with the distribution function $F$.
Then, by the Portmanteau theorem, $\hat\mu \weakto \mu$ a.s.

The supports of the empirical CDFs $\hat F_1,\hat F_2,\hat F_1^w, \hat F_2^w,\hat F^w$ lie in the respective supports of $F_1$, $F_2$, $F_1^w$, $F_2^w$, $F^w$, which are bounded by \Cref{a:bounded-support}.
Therefore, $\hat \mu$ also converges to $\mu$ in $L^1$, a.s.

Because weak convergence and convergence in $L^1$ imply convergence in the Wasserstein distance $W_1$ (see, e.g., Theorem 6.9 in \citet{villani2009optimal}), we conclude that $W_1(\hat{\mu},\mu)\to 0 \text{ a.s.}$

By the Kantorovich–Rubinstein duality, for any measures $\mu,\nu$ and any $K>0$,
\begin{align*}
    W_{1}(\mu ,\nu )= \frac {1}{K}\sup_{\|f\|_{Lip}\leq K} \left( \int f(z) \, d\mu(z) - \int f(z)\, d\nu(z) \right).
\end{align*}
Applying this representation to measures $\hat{\mu}$ and  $\mu$ and the moment function $m(z;\theta)$ componentwise and using \Cref{a:bounded-support},
\begin{equation}
\frac {1}{M_1} \sup_{\theta\in\Theta}\left\vert \int  m(z;\theta)d\hat{F} - \int  m(z;\theta)dF\right\vert \le W_{1}(\hat{\mu},\mu).
\label{ap:B}    
\end{equation}
Taking into account that $W_1(\hat\mu,\mu) \to 0$ a.s. completes the proof.

\section{Proof of \Cref{lem:ULLN-prime}}\label{proof_lem_prime}

We follow the proof of Lemma 1 in \citet{tauchen1985diagnostic}. \Cref{a:m-Lipschitz}.\ref{m-dominated} ensures that the expectation $M(\theta):=\int m(z;\theta)\, dF(z)$ exists.
\Cref{a:m-Lipschitz}.\ref{m-cont-in-theta} implies by dominated convergence that $M(\theta)$ is a continuous function of $\theta$.  

Fix any $\varepsilon >0$.
By almost sure continuity, $\lim u(Z;\theta,d)=0$ as $d\to 0$, with $\theta$ fixed, $F$-a.s.
Thus by the dominated convergence theorem, there exists $\bar{d}(\theta)$ such that $\E[u(Z;\theta,d)]\le \varepsilon$ whenever $d\le \bar{d}(\theta)$.
Let $B(\theta)$ denote the open ball of radius $\bar{d}(\theta)$ around each $\theta\in\Theta$. Together the balls $B(\theta)$, $\theta\in\Theta$ cover $\Theta$.
By compactness of $\Theta$, this cover has a finite subcover $B_k = B(\theta_k)$, $k=1,\ldots, K$.
Denote $d_k = \bar{d}(\theta_k)$ and $\mu_k = \E[u(Z;\theta_k,d_k)]$.
Note that if $\theta\in B_k,$ then $\mu_k\le \varepsilon$, and $\vert M(\theta) - M(\theta_k)\vert\le\varepsilon$.
For any $\theta\in B_k$,
\begin{align*}
    &\left\vert \int m(z;\theta)\, d\hat{F}(z) - \int m(z;\theta)\,dF(z)\right\vert  \\
    &\le \left\vert \int m(z;\theta)\, d\hat{F}(z) - \int m(z;\theta_k)\,d\hat{F}(z)\right\vert+\left\vert \int m(z;\theta_k)\, d\hat{F}(z) - M(\theta_k)\right\vert +\underbrace{\left\vert M(\theta_k)-M(\theta)\right\vert}_{\le\epsilon}\\
    &\le \int |m(z;\theta)-m(z;\theta_k)|\, d\hat{F}(z) +\left\vert \int m(z;\theta_k)\, d\hat{F}(z) - \int m(z;\theta_k)\,dF(z)\right\vert +\varepsilon\\
    &\le \int u(z;\theta_k,d_k)\, d\hat{F}(z) +\left\vert \int m(z;\theta_k)\, d\hat{F}(z) - \int m(z;\theta_k)\,dF(z)\right\vert +\varepsilon\\
    &\le \underbrace{\left\vert\int u(z;\theta_k,d_k)\, d\hat{F}(z)-\int u(z;\theta_k,d_k)\,dF(z)\right\vert}_{I}+\underbrace{\mu_k}_{\le\varepsilon} +\underbrace{\left\vert \int m(z;\theta_k)\, d\hat{F}(z) - \int m(z;\theta_k)\,dF(z)\right\vert}_{II} +\varepsilon
\end{align*}
Consider the functional
\begin{align*}
    T_\varphi(F) = \int \varphi(z) \, dF(z)
\end{align*}
for functions of bounded variation $\varphi: \R^2 \to \R$ and $F: \R^2 \to \R$.
By multivariate integration by parts (see, e.g., \citet{young1917multiple}),
\begin{align}\label{a:by_parts}
\int_{x,y}^{a,b} \varphi(z_1,z_2) \, dF(z_1,z_2) 
&=  \varphi(x,y)\left[F\right]_{x,y}^{a,b}
+\int_{x,y}^{a,b} \left[F\right]_{z_1,z_2}^{a,b} d \varphi(z_1,z_2) \notag \\
&+\int_{x}^{a} \left[F\right]_{z_1,y}^{a,b} d \varphi(z_1,y)+\int_{y}^{b}\left[F\right]_{x,z_2}^{a,b} d \varphi(x,z_2), 
\end{align}    
where
\[
\left[F\right]_{x,y}^{a,b} := F(x,y)-F(a,y)-F(x,b)+F(a,b).
\]
An analogous formula is also valid for any finite dimension.
This implies that $T$ is continuous in the uniform norm.

Finally, let us show that terms $I$ and $II$ can be made arbitrarily small.
Since both functions $u(z;\theta,d) $ and $m(z;\theta)$ are of bounded variation for fixed $\theta$ and $d$ by \Cref{a:m-Lipschitz}.\ref{m-bv} and \ref{a:m-Lipschitz}.\ref{u-bv}, we can use the multivariate version of the integration by parts formula \eqref{a:by_parts}.
Since $T$ is continuous and $\hat{F}$ converges a.s. to $F$ in the uniform norm (see \cref{eq:F-hat-to-F}), both terms I and II can be made arbitrarily small for all $k=1,\dots, K$, completing the proof.

\section{Proof of \Cref{thm:inference}}\label{proof_inference}

Denote the observed vector $\eta = (W,Z_1',W Z_2',(Z_2^r)')' \in \R^{1+3d}$ and let $F_\eta$ be its CDF.
Let $\hat F_\eta$ be the corresponding empirical CDF,
\begin{align*}
    \hat F_\eta(w,z_1,z_2,z_2^r) = \frac{1}{n} \sumin 1\left(W_i\le w, Z_{1i}\le z_1, W_iZ_{2i}\le z_2, Z_{2i}^r \le z_2^r \right).
\end{align*}
We can recover $p, F_1,F_2,F_1^w,F_2^w,F^w$ as functionals of $F_\eta \in \ell^\infty(\R^{1+3d})$ to corresponding spaces of bounded functions.
These functionals are defined by the following formulas.
\begin{align*}
 p(F_\eta)& = 1 - F_\eta (0,\infty,\infty,\infty), \\   
 F_1(F_\eta)(z_1) &= F_\eta (1,z_1,\infty,\infty), \\
 F_2(F_\eta)(z_2) &= F_\eta(1,\infty,\infty,z_2), \\
F_1^w(F_\eta)(z_1) &= \Prb(Z_1\le z_1|W=1) = \frac{\Prb(Z_1\le z_1,W=1)}{\Prb(W=1)} = \frac{\Prb(Z_1\le z_1,W\le 1)- \Prb(Z_1\le z_1,W\le 0)}{\Prb(W\le 1)-\Prb(W\le 0)} \\
&= \frac{F_\eta(1,z_1,\infty,\infty)-F_\eta(0,z_1,\infty,\infty) }{F_\eta(1,\infty,\infty,\infty)-F_\eta(0,\infty,\infty,\infty)}, \\
F_2^w(F_\eta)(z_2) &= \Prb(Z_2\le z_2|W=1) = \frac{\Prb(W Z_2\le z_2,W=1)}{\Prb(W=1)} \\
&= \frac{\Prb(WZ_2\le z_2,W\le 1)- \Prb(WZ_2\le z_2,W\le 0)}{\Prb(W\le 1)-\Prb(W\le 0)} = \frac{F_\eta(1,\infty,z_2,\infty)-F_\eta(0,\infty,z_2,\infty) }{F_\eta(1,\infty,\infty,\infty)-F_\eta(0,\infty,\infty,\infty)}, \\
F^w(F_\eta)(z_1,z_2) &= \Prb(Z_1\le z_1,Z_2\le z_2|W=1) = \frac{\Prb(Z_1\le z_1, W Z_2\le z_2,W=1)}{\Prb(W=1)} \\
& = \frac{\Prb(Z_1\le z_1, WZ_2\le z_2,W\le 1)- \Prb(Z_1\le z_1, WZ_2\le z_2,W\le 0)}{\Prb(W\le 1)-\Prb(W\le 0)} \\
&= \frac{F_\eta(1,z_1,z_2,\infty)-F_\eta(0,z_1,z_2,\infty)}{F_\eta(1,\infty,\infty,\infty)-F_\eta(0,\infty,\infty,\infty)}.
\end{align*}
Given these functionals, the target CDF $F$ also becomes a functional of $F_\eta$ via
\begin{align}\label{small_phi}
    F = \phi(F_\eta) :=  \Phi(p(F_\eta), F_1(F_\eta), F_2(F_\eta),F_1^w(F_\eta),F_2^w(F_\eta),F^w(F_\eta)).
\end{align}

\textbf{Step 1} (Hadamard differentiability of $F=\phi(F_\eta)$).

Let us show that $\phi: \ell^\infty(\R^{1+3d}) \to \ell^\infty(\R^{2d})$ is Hadamard differentiable at the true CDF $F_\eta$.
Here $\ell^\infty(\mathcal{A})$ is the space of bounded functions on the set $\mathcal{A}$.

Since $\Phi$ is differentiable, it suffices to show that the functionals $p, F_1,F_2,F_1^w,F_2^w,F^w$ are Hadamard differentiable, see Lemma 3.10.26 in \citet{vaart2023empirical}.

The functionals $F_1$, $F_2$, and $p$ are Hadamard differentiable because they are linear and bounded.

Let $\eta=(w,z_1,z_2,z_2^r)$. The functional $F_1^w$ is Hadamard differentiable at the true $F_\eta$ with the derivative
\begin{align*}
    (F_1^w)'(H)(\eta) = \frac{(F_1-F_0)(H_{1,z_1}-H_{0,z_1}) - (H_1-H_0)(F_{1,z_1} - F_{0,z_1})}{(F_1-F_0)^2},
\end{align*}
where $H \in \ell^\infty(\R^{1+3d})$ is an arbitrary direction and
\begin{align*}
   F_1 &= F_{\eta}(1,\infty,\infty,\infty), & F_0 &= F_{\eta}(0,\infty,\infty,\infty),\\
   F_{1,z_1} &= F_{\eta}(1,z_1,\infty,\infty), &  F_{0,z_1} &= F_{\eta}(0,z_1,\infty,\infty),\\
   H_1&= H(1,\infty,\infty,\infty), & H_0 &= H(0,\infty,\infty,\infty),\\
   H_{1,z_1} &= H(1,z_1,\infty,\infty), &  H_{0,z_1} &= H(0,z_1,\infty,\infty).
\end{align*}
Similarly, the Hadamard derivatives of $F_2^w$ and $F^w$, respectively, are
\begin{align*}
    (F_2^w)'(H)(\eta) &= \frac{(F_1-F_0)(H_{1,z_2}-H_{0,z_2}) - (H_1-H_0)(F_{1,z_2} - F_{0,z_2})}{(F_1-F_0)^2}, \\
    (F^w)'(H)(\eta) &= \frac{(F_1-F_0)(H_{1,z_1,z_2}-H_{0,z_1,z_2}) - (H_1-H_0)(F_{1,z_1,z_2} - F_{0,z_1,z_2})}{(F_1-F_0)^2},
\end{align*}
where 
\begin{align*}
   F_{1,z_2} &= F_{\eta}(1,\infty,z_2,\infty), &  F_{0,z_2} &= F_{\eta}(0,\infty,z_2,\infty),\\
   F_{1,z_1,z_2} &= F_{\eta}(1,z_1,z_2,\infty), &  F_{0,z_1,z_2} &= F_{\eta}(0,z_1,z_2,\infty),\\
   H_{1,z_2} &= H(1,\infty,z_2,\infty), &  H_{0,z_2} &= H(0,\infty,z_2,\infty),\\
   H_{1,z_1,z_2} &= H(1,z_1,z_2,\infty), &  H_{0,z_1,z_2} &= H(0,z_1,z_2,\infty).
\end{align*}
Note that $F_1 - F_0 \equiv p > 0$ by \Cref{a:identif}.

\textbf{Step 2} (Hadamard differentiability of $\theta=\theta(F)$).

Denote the functional
\begin{align}
    M(\theta,F) := \int m(z;\theta) \, d F(z)
\end{align}
and let $\theta(F)$ be its unique root, i.e.
\begin{align}
    M(\theta(F),F) = 0.
\end{align}
Let us show that $\theta$ is Fr\'{e}chet (and hence Hadamard) differentiable at $F_0$, the true value of $F$.
In what follows, differentiability is understood in the Fr\'{e}chet sense.
For this, we check the conditions of the implicit function theorem in Banach spaces; see, e.g., Theorem 5.9 in \citet{lang2012fundamentals}:
\begin{enumerate}
    \item $M(\theta,F)$ is continuously differentiable on $\Theta_0 \times \mathcal{F}_0$
    \item $D_1 M(\theta_0,F_0): \Theta \to \R^{d_\theta}$ is a \emph{toplinear isomorphism}, i.e., a bounded linear map with a bounded inverse.
\end{enumerate}

First, $M$ is continuously differentiable in $\theta \in \Theta_0$ for all $F \in \mathcal{F}_0$ by \Cref{as:M-is-differentiable}.\ref{as:M-diff-in-theta}.
Second, $M(\theta,F)$ is differentiable in $F$ because it is a bounded linear functional of $F$.
Since the existence of partial derivatives suffices for differentiability (see, e.g., Proposition 3.5 in \citet{lang2012fundamentals}), this establishes condition 1.

Condition 2 follows directly from \Cref{as:M-is-differentiable}.\ref{as:M-diff-nondegenerate} because for a derivative w.r.t. a finite-dimensional parameter $\theta$, its existence and nondegeneracy are equivalent to being a toplinear isomorphism.

\textbf{Step 3} (asymptotic normality and bootstrap validity).

The standard Donsker theorem states that
\begin{align}
    \sqrt{n} \left( \hat F_\eta - F_\eta \right) \weakto \G_{F_\eta} \text{ in } \ell^\infty(\R^{1+3d}), \label{eq:donsker-F-eta}
\end{align}
where $\G_{F_\eta}$ is the $F_\eta$-Brownian bridge, see, e.g., Example 2.5.4 in \citet{vaart2023empirical}. Moreover, since $\hat F_\eta$ is an empirical CDF, the nonparametric bootstrap is valid for it in the sense of Theorem 3.7.1 in \citet{vaart2023empirical}.

Applying Theorem 3.10.11 in \citet{vaart2023empirical} to the Hadamard differentiable functional $\theta=\theta(F)=\theta(\phi(F_\eta))$ then implies that the asymptotic normality and validity of the nonparametric bootstrap for $\hat\theta$.

 \section{Simulation results under misspecified link function}\label{sec:misspec}

    Here we present the results of Monte Carlo simulations performed under a misspecified link function.
    We report the results for our discrete DGP, see Section \ref{sec:mc}; the results for the continuous DGP are qualitatively similar.
    
    We maintain the logistic CDF $G(x)=(1+e^{-x})^{-1}$ (which has $\pi^2/3 \approx 3.29$) as the true link function, while choosing the zero-mean Gaussian link function for estimation.
    Notice that the variance of the latter can be set to an arbitrary positive value since it is subsumed by the nuisance functions $k_1$ and $k_2$.
    This is illustrated by the fact that the simulation results under the true variance in Table \ref{tab-mc-discr-misspec0} and those under the unit variance in Table \ref{tab-mc-discr-misspec1} are identical.

    Comparing Tables \ref{tab-mc-discr} and \ref{tab-mc-discr-misspec0} shows that using the Gaussian link in place of the true logistic link does not lead to any noticeable deterioration in the performance of our estimation and inference procedures.
    While some other (carefully selected) directions of misspecification may lead to such deterioration, we believe that leaving $k_1,k_2$ unrestricted gives our estimator sufficient adaptability to the underlying link function.

\begin{table}[h!]
\centering
\small
\begin{tabular}{c c c c c c }
& & \multicolumn{2}{c}{$n_1=n_r=1000$} & \multicolumn{2}{c}{$n_1=n_r=10,000$}\\											
\cmidrule(lr){3-4} \cmidrule(lr){5-6}											
	&		&	$\hat\theta$	&	$\hat\theta_{naive}$	&	$\hat\theta$	&	$\hat\theta_{naive}$	\\
    \hline
\multirow{6}{*}{$m=5$}	&	bias	&	-0.005	&	-0.208	&	-0.002	&	-0.205	\\
	&	rmse	&	0.229	&	0.279	&	0.071	&	0.213	\\
	&	mae	&	0.183	&	0.236	&	0.056	&	0.205	\\
	&	coverage $99\%$	&	0.993	&		&	0.992	&		\\
	&	coverage $95\%$	&	0.953	&		&	0.953	&		\\
	&	coverage $90\%$	&	0.905	&		&	0.903	&		\\
    \hline
\multirow{6}{*}{$m=10$}	&	bias	&	0.003	&	-0.171	&	-0.002	&	-0.171	\\
	&	rmse	&	0.494	&	0.427	&	0.150	&	0.210	\\
	&	mae	&	0.387	&	0.351	&	0.120	&	0.181	\\
	&	coverage $99\%$	&	0.998	&		&	0.993	&		\\
	&	coverage $95\%$	&	0.964	&		&	0.957	&		\\
	&	coverage $90\%$	&	0.907	&		&	0.908	&		\\
    \hline
\multirow{6}{*}{$m=20$}	&	bias	&	-0.004	&	-0.030	&	0.002	&	-0.027	\\
	&	rmse	&	0.881	&	0.836	&	0.271	&	0.264	\\
	&	mae	&	0.688	&	0.680	&	0.215	&	0.212	\\
	&	coverage $99\%$	&	0.998	&		&	0.9902	&		\\
	&	coverage $95\%$	&	0.981	&		&	0.949	&		\\
	&	coverage $90\%$	&	0.952	&		&	0.903	&		
\end{tabular}
\caption{Simulation results for the discrete DGP under a misspecified Gaussian link function with $V=\pi^2/3 \approx 3.29$. Number of simulations $S=5,000$. Bias, RMSE, and MAE are reported as shares of the true value of $\theta(m)$}
\label{tab-mc-discr-misspec0}
\end{table}

\begin{table}[h!]
\centering
\small
\begin{tabular}{c c c c c c }
& & \multicolumn{2}{c}{$n_1=n_r=1000$} & \multicolumn{2}{c}{$n_1=n_r=10,000$}\\				
\cmidrule(lr){3-4} \cmidrule(lr){5-6}											
	&		&	$\hat\theta$	&	$\hat\theta_{naive}$	&	$\hat\theta$	&	$\hat\theta_{naive}$	\\
    \hline
    \multirow{6}{*}{$m=5$}	&	bias	&	-0.005	&	-0.208	&	-0.002	&	-0.205	\\
	&	rmse	&	0.229	&	0.279	&	0.071	&	0.213	\\
	&	mae	&	0.183	&	0.236	&	0.056	&	0.205	\\
	&	coverage $99\%$	&	0.993	&		&	0.992	&		\\
	&	coverage $95\%$	&	0.953	&		&	0.953	&		\\
	&	coverage $90\%$	&	0.905	&		&	0.903	&		\\
    \hline
\multirow{6}{*}{$m=10$}	&	bias	&	0.003	&	-0.171	&	-0.003	&	-0.170	\\
	&	rmse	&	0.494	&	0.427	&	0.147	&	0.209	\\
	&	mae	&	0.387	&	0.351	&	0.118	&	0.180	\\
	&	coverage $99\%$	&	0.998	&		&	0.995	&		\\
	&	coverage $95\%$	&	0.964	&		&	0.955	&		\\
	&	coverage $90\%$	&	0.907	&		&	0.910	&		\\
    \hline
\multirow{6}{*}{$m=20$}	&	bias	&	-0.004	&	-0.030	&	0.000	&	-0.029	\\
	&	rmse	&	0.881	&	0.836	&	0.268	&	0.262	\\
	&	mae	&	0.688	&	0.680	&	0.211	&	0.209	\\
	&	coverage $99\%$	&	0.998	&		&	0.983	&		\\
	&	coverage $95\%$	&	0.981	&		&	0.945	&		\\
	&	coverage $90\%$	&	0.952	&		&	0.900	&		
\end{tabular}
\caption{Simulation results for the discrete DGP under a misspecified Gaussian link function with variance $V=1$. Number of simulations $S=5,000$. Bias, RMSE, and MAE are reported as shares of the true value of $\theta(m)$.}
\label{tab-mc-discr-misspec1}
\end{table}

\end{subappendices}

\end{document}